\documentclass[10pt,conference,compsocconf,letterpaper]{IEEEtran}

\usepackage[ruled,linesnumbered]{algorithm2e}
\usepackage{enumerate}
\usepackage{epsfig}
\usepackage{graphicx}
\usepackage{amssymb}
\usepackage{url}
\usepackage{subfigure}
\usepackage{wrapfig}

\newtheorem{thm}{Theorem}

\begin{document}

\title{Fit and Vulnerable: Attacks and Defenses for a Health Monitoring Device}

\author{
\IEEEauthorblockN{Mahmudur Rahman, Bogdan Carbunar, Madhusudan Banik}
\IEEEauthorblockA{School of Computing and Information Sciences\\
Florida International University\\
Miami, Florida 33199\\
Email: \{mrahm004, carbunar, mbani002\}@cs.fiu.edu}
}

\maketitle

\begin{abstract}

The fusion of social networks and wearable sensors is becoming increasingly
popular, with systems like Fitbit automating the process of reporting and
sharing user fitness data. In this paper we show that while compelling, the
careless integration of health data into social networks is fraught with
privacy and security vulnerabilities. Case in point, by reverse engineering the
communication protocol, storage details and operation codes, we identified
several vulnerabilities in Fitbit. We have built FitBite, a suite of tools
that exploit these vulnerabilities to launch a wide range of attacks against
Fitbit. Besides eavesdropping, injection and denial of service, several attacks
can lead to rewards and financial gains. We have built FitLock, a lightweight
defense system that protects Fitbit while imposing only a small overhead. Our
experiments on BeagleBoard and Xperia devices show that FitLock's end-to-end
overhead over Fitbit is only 2.4\%.

\end{abstract}

\section{Introduction}

Online social networks are sites that enable their users to connect and share
information with friends and family. Recent advances in wearable, user-friendly
devices equipped with smart sensors (e.g., pedometers, heart rate and sleep
monitors) and wireless technologies, are facilitating the emergence of
\textit{social sensor networks} (SSNs): social networks that collect and share
not only explicit user information (e.g., status updates, location reports) but
also implicit health-centric data.

Fitbit~\cite{Fitbit}, a representative \textit{social sensor network} centers
its existence on fitness sensor data. It consists of (i) \textit{trackers},
wireless-enabled, wearable devices that record their users' daily step counts,
distance traversed, calories burned and floors climbed as well as sleep
patterns when worn during the night and (ii) an online social network that
automatically captures, displays and shares fitness data of its users.
Figure~\ref{fig:trend:all} illustrates the basic functionality of Fitbit.

While popular and useful in its encouragement of healthy lifestyles, the
combination of health sensors and social networks makes social sensor networks
the source of significant privacy and security issues. In this paper we show
that Fitbit is vulnerable to a wide range of attacks. Besides standard social
networking problems, including infiltration attacks~\cite{BMBR11} and private
data leaks to general account holders~\footnote{Fitbit has suffered
criticism due to its initial default access control settings: The reported
sensor information was made publicly available on Fitbit's social network.},
Fitbit is made vulnerable by the wireless nature of tracker communications and
poor security practices. 

Fitbit relies on a Personal Area Network (PAN) protocol~\cite{PAN} called
ANT~\cite{ANTref}, that enables trackers to automatically upload their data to
the online social network account of their user. The improper design of
Fitbit's communication protocol, (i) allows Fitbit users to engineer their
fitness data and inject it into their social networking accounts, thus gain
financial benefits and (ii) enables external attackers to intercept data
reported by trackers of other users, inject arbitrary data into the trackers
and online social network accounts of other users, as well as launch denial of
service attacks.

\begin{figure}
\centering
\includegraphics[width=2.2in]{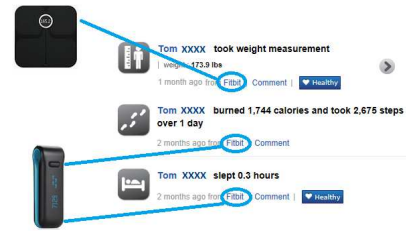}
\caption{Social Sensor Network (SSN) illustration. Health sensor devices and
corresponding data reported, displayed and shared with friends. The user's last
name is anonymized.
\label{fig:trend:all}}
\vspace{-15pt}
\end{figure}

In order to expose Fitbit's vulnerabilities, in a first contribution, we have
reverse engineered the semantics of tracker memory banks, the command types and
the tracker-to-social network communication protocol.  In a second
contribution, we have built FitBite, a suite of tools that exploit Fitbit's
faulty design. We have used FitBite to prove the feasibility of a wide range of
attacks. For instance, we show that FitBite allows attackers to capture and
modify the data stored on any tracker situated within a radius of 15 ft.

%

In a third contribution, we propose FitLock, a lightweight extension that uses
efficient cryptographic tools to secure the Fitbit protocol. We show that
FitLock prevents the FitBite attacks.  Our end-to-end implementation on
BeagleBoard~\cite{Beagle} and Xperia devices shows that the computation and
communication overhead imposed by FitLock is small: resource constrained
BeagleBoard and Xperia devices can support hundreds of packet encryption and
transmission operations per second. Moreover, the end-to-end overhead of the
cryptographic operations employed by FitLock over the standard Fitbit protocol
is only 2.4\% on a Xperia device.  The project website containing the source
code of FitBite and FitLock is made publicly available at ~\cite{Fitlock}.

The paper is organized as follows. Section~\ref{sec:model} describes the Fitbit
model, including background on the ANT protocol, and details the attacker model
considered. Section~\ref{sec:fitbit} reverse engineers Fitbit and
Section~\ref{sec:fitbite} introduces FitBite, the suite of attack tools.
Section~\ref{sec:fitlock} introduces FitLock, our defense extension and proves
its security. Section~\ref{sec:evaluation} describes our implementation
results. Section~\ref{sec:related} describes related work and
Section~\ref{sec:conclusions} concludes.

\section{Background and Model}
\label{sec:model}

The Fitbit system consists of user tracker devices, user USB base stations and
an online social network.

\noindent
{\bf Fitbit tracker.}
The Fitbit \textit{tracker} is a wearable device that relies on a 3D motion
sensor and a barometric pressure sensor to measure the daily steps taken,
distance traveled, floors climbed, calories burned, and the duration and
intensity of the user exercise. The tracker mainly consists of four IC chips,
(i) a MMA7341L 3-axis MEMS accelerometer, (ii) a MSP430F2618 low power TI MCU
consisting of 92 KB of flash and 96 KB of RAM, (iii) a nRF24API 2.4 GHz RF chip
supporting the ANT protocol (1 Mbits/sec, 15 ft transmission range) and (iv) a
MEMS altimeter to count the number of floors climbed. The user can switch
between displaying different real-time fitness information on the tracker, using
a dedicated hardware \textit{switch} button (see the arrow pointing to the
switch in Figure~\ref{fig:system:main}).  Each tracker has a unique id, called
the \textit{tracker public id} (TPI).


\noindent
{\bf Data conversion.}
The accelerometer and the altimeter allow the tracker to count the steps taken
and the floors climbed by the user. The tracker relies on extrapolated walk/run
stride length values to convert the step count into the distance covered by the
user: the sum of the recorded walking steps times the user walking stride
length and of the running steps times the user running stride length.  The
running steps are identified based on the frequency and intensity of the user's
steps. The tracker uses the extrapolated user Basal Metabolic Rate
(BMR)~\cite{BMR} values to convert the user's daily activities into the number
of calories burned. 




\begin{figure}
\begin{center}
\includegraphics[width=2.9in,height=1.8in]{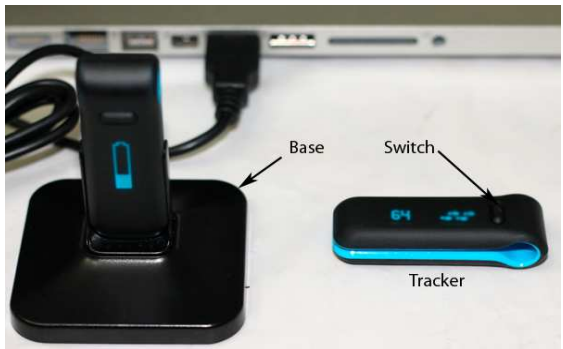}
\caption{Fitbit system components: trackers (one mounted on the base), the base
(arrow indicated), user laptop. The arrow pointing to the tracker shows the
switch button, allowing the user to display various fitness data.
\label{fig:system:main}}
\end{center}
\vspace{-15pt}
\end{figure}

\noindent
{\bf The base.}
The second component, the Fitbit \textit{base}, connects to the user's main
compute center (e.g., PC, laptop) and is equipped with a wireless communication
chip that enables it to communicate with any tracker within a range of 15 ft.
The base acts as a bridge between trackers and the online social network. It
sets up connections with all the trackers within its transmission range, then
reads and clears up the information stored on the tracker according to commands
issued by the social network. Figure~\ref{fig:system:main} shows a snapshot of
two trackers and a base, connected to a laptop through a USB port.

\noindent
{\bf The webserver.}
The third component, the online social network, allows users to create accounts
from which they befriend and maintain contact with other users. Upon purchase
of a Fitbit tracker and base, a user binds the tracker to her social network
account. Each social network account has a unique id, called the
\textit{user public id} (UPI). After the base detects and sets up
a connection with the tracker, the base automatically collects and reports
tracker stored information (step count, distance, calories, sleep patterns) to
the corresponding social network account. Based on user preferences, this data
is either made public, shared with the user's friends, or kept private. In the
following, we use the term \textit{webserver} to denote the computing resources
of the online social network.

\noindent
{\bf Tracker-to-base communication: the ANT protocol.}
Trackers communicate to bases over ANT. ANT is a 2.4 GHz bidirectional wireless
Personal Area Network (PAN) communications technology optimized for
transferring low-data rate, low-latency data between multiple ANT-enabled
devices.  The ultra-low power consumption of the ANT chipset guarantees an
extended battery life even from low-capacity supplies such as a coin cell
battery. This, along with the low implementation cost of ANT, enables its
integration and use in a wide range of mobile devices, including smartphones
and health sensors.

\subsection{Attacker Model}

We assume the existence of not only external attackers but also insiders.
External attackers attempt to learn and modify the fitness information reported
by the trackers of other users, as well as disrupt the Fitbit protocol.
Insiders own Fitbit trackers and may attempt to report fitness values that do
not reflect their effort, e.g., inflate reports or replay old values.  We
assume external attackers do not have physical access to trackers of other
users. However, attackers are able to capture wireless communications in their
vicinity. Furthermore, we assume that the Fitbit service (e.g. the social
network servers) does not collude with attackers to facilitate false data
reports.

\section{Reverse Engineering Fitbit}
\label{sec:fitbit}

We document here the results of our effort to reverse engineer the Fitbit
communication protocol, including the message communication format among 
the participating devices. Our endeavor has relied on information from
libfitbit~\cite{libfitbit} for open source health hardware access.

Fitbit uses \textit{service logs}, files that store information concerning 
communications involving the base. On the Windows installation of the Fitbit
software, daily logs are stored in cleartext
%
%
in files whose names record the hour, minute and second corresponding to the 
time of the first log occurrence. Each request and response involving the 
tracker, base and social network is logged and sometimes even documented in the
archive folder of that log directory. The logs have proved central to our
understanding and reverse engineering of the functionality of Fitbit.

Data retrieved from the tracker to be uploaded to the social network is encoded
in base64 format. However, no authentication is used, and all requests are sent
in clear HTTP.  We have exploited Fitbit's lack of encryption in the messages
sent between the base and the tracker to implement a USB based filter driver
that separately logs the data flowing to and from the base. 

In the following we present details of the organization of the tracker memory
banks, then describe the main Fitbit opcodes and then present the reverse
engineered Fitbit communication protocol.

\subsection{Memory Banks}

A tracker has two types of memory banks, (i) \textit{read banks}, containing
data to be read by the base and (ii) \textit{write banks}, containing data that
can be written by the base. The log data we captured reveals that during the 
upload session, the webserver reads data from 6 memory banks, writes on 2 write 
memory banks and clears data from 5 memory banks by sending requests to the 
tracker through the base. The byte length of memory banks varies. We now 
briefly describe the most important read and write memory banks.





\noindent
\textbf{Read banks}.
The read bank \#1 stores the daily user fitness records. Each record is 16
bytes long. It starts with a 4 byte long timestamp, followed by the number of
calories, steps, distance and floor count. Both steps and distance are stored
on four bytes while the calories and the floor count are stored on two bytes.
To ensure reliability, Fitbit stores the important fitness records (e.g., step
and floor count) on multiple memory banks.

\noindent
\textbf{Write banks.}
The write bank \#0 stores 64 bytes concerning the device settings as specified
on the user's Fitbit account (the ``Device Settings'' and ``Profile Settings''
links). The write bank \#1 stores 16 bytes that contain the daily user
fitness records whose data format is similar to the read memory bank 0.



\subsection{Opcodes and Responses}

The webserver communicates with a tracker through a base. The communication is
embedded in XML blocks, that contain base64 encoded opcodes -- commands for the
tracker. Opcodes are 7 bytes long. We briefly list below the most important
opcodes and their corresponding responses. The opcode types are also shown
in Figure~\ref{fig:fitbit:maincomm}.

\noindent
{\bf Retrieve device information (TRQ-REQ):} opcode $[0x24,0 (6 times)]$.
Upon receiving this opcode from the webserver (via the base), the tracker
includes in a reply its serial number (5 bytes), the hardware revision number,
%
%
and whether the tracker is plugged in on the base.
%

\noindent
{\bf Read/write tracker memory.}
To read a memory bank, the webserver needs to issue the READ-TRQ opcode,
$[0x22, index,0 (5 times)]$, where $index$ denotes the memory bank requested.
The response embeds the content of the specified memory bank. To write data to
a memory bank, the webserver issues the WRITE opcode $[0x23, index,
datalength,0 (4 times)]$.  The payload data is sent along with the opcode. The
value $index$ denotes the destination memory bank and $datalen$ is the length
of the payload. A successful operation returns the response $[0x41,0 (6
times)]$.

\noindent
{\bf Erase memory: (ERASE)} opcode $[0x25, index, time, 0]$.  The
webserver specifies the $index$ denoting the memory bank to be erased.
$time$ (4 bytes, MSB) is the operation deadline - the date until which the
data should be erased. A successful operation returns the response $[0x41,0(6
times)]$.

\subsection{The Fitbit Communication Protocol}

\begin{figure}
\begin{center}
\includegraphics[width=3.3in]{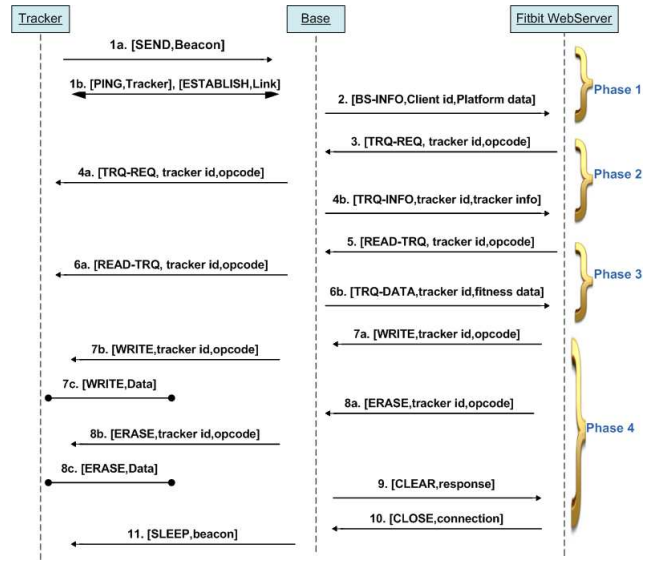}
\caption{Fitbit protocol between the tracker, base and the Fitbit webserver}
\label{fig:fitbit:maincomm}
\end{center}
\vspace{-15pt}
\end{figure}

In the following, for brevity, we use the notation ``URL'' to 
denote the full URL \url{http://client.fitbit.com}. The data flow between 
the tracker, base and the webserver during the data upload operation, 
illustrated in Figure~\ref{fig:fitbit:maincomm}, is divided into 4
phases, beginning at steps 2, 3, 5 and 7:

\begin{enumerate}

\item
Upon receiving a beacon from the tracker, the base establishes a connection
with the tracker.

\item
{\bf Phase 1:}
The base contacts the webserver at the \url{URL/device/tracker/uploadData} 
and sends basic client and platform information.

\item
{\bf Phase 2:}
The webserver sends the tracker id and the opcode for retrieving tracker
information (TRQ-REQ).

\item
The base contacts the specified tracker, retrieves its information TRQ-INFO
(serial number, firmware version, etc.) and sends it to the webserver at the
\url{URL/device/tracker/dumpData/lookupTracker}.

\item
{\bf Phase 3:}
Given the tracker's serial number, the webserver retrieves the associated
tracker public id (TPI) and user public id (UPI) values. The webserver sends to
the base the TPI/UPI values along with the opcodes for retrieving fitness
data from the tracker (READ-TRQ).


\item
The base forwards the TPI and UPI values and the opcodes to the tracker,
retrieves the fitness data from the tracker (TRQ-DATA) and sends it to the
webserver at the \url{URL/device/tracker/dumpData/dumpData}.

\item
{\bf Phase 4:}
The webserver sends to the base opcodes to WRITE updates provided by the user
in her Fitbit social network account (device and profile settings, e.g., body
and personal information, time zone, etc). This operation takes place
irrespective of whether the user has updated her settings since the last
communication with the tracker or not. The base forwards the WRITE opcode and
the updates to the tracker, who overwrites the previous values on its write
memory banks.

\item
The webserver sends opcodes to ERASE the fitness data from the tracker. The
base forwards the ERASE request to the tracker, who then erases the contents of
the corresponding read memory banks.


\item
The base forwards the response codes for the executed opcodes from the tracker to
the webserver at the \url{URL/device/tracker/dumpData/clearDataConfigTracker}.

\item
The webserver replies to the base with the opcode to CLOSE the tracker.

\item
The base requests the tracker to SLEEP for 15 minutes, before sending its next
beacon.

\end{enumerate}

\begin{figure}
\centering
\vspace{-5pt}
\includegraphics[width=3.1in]{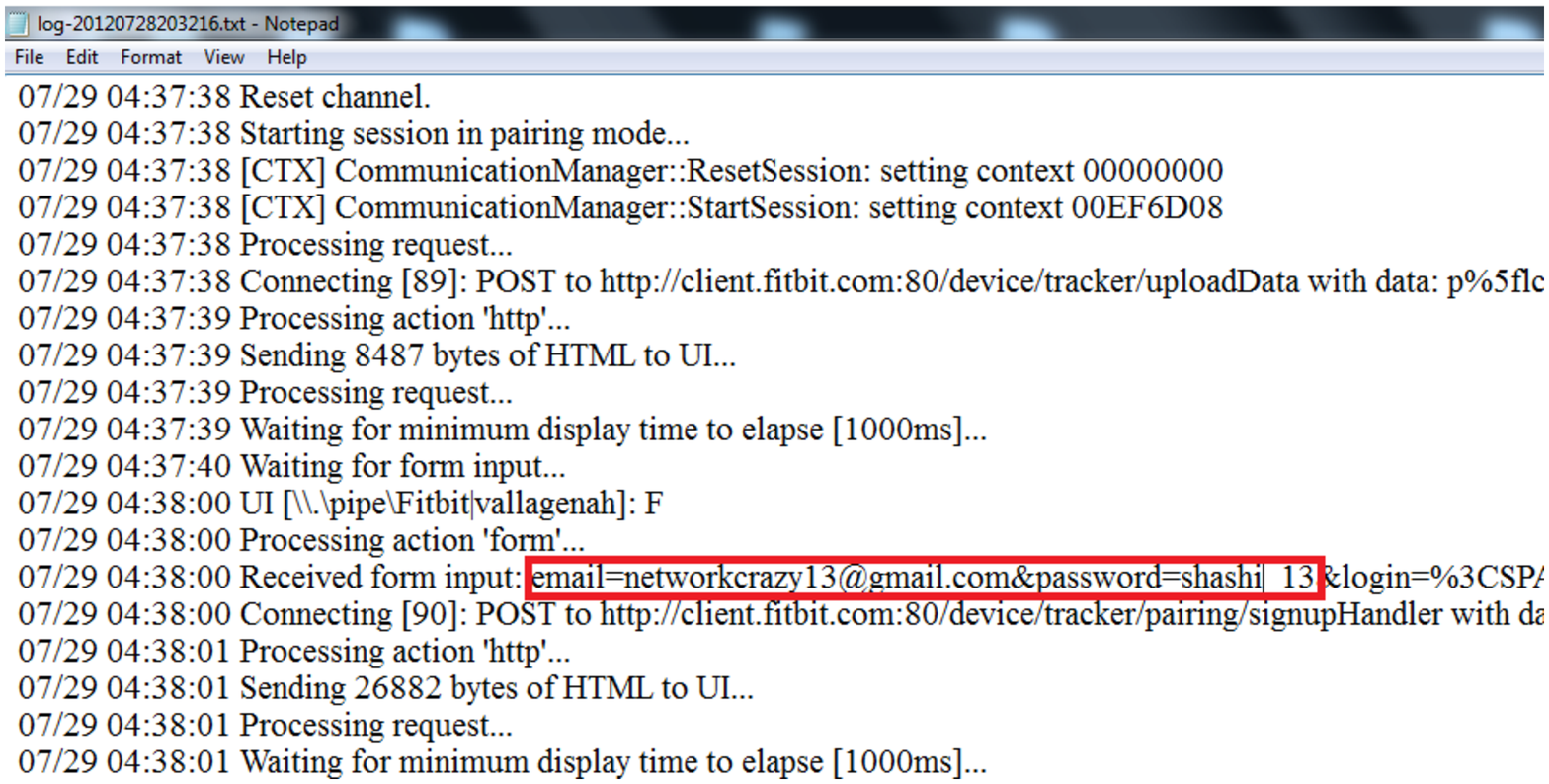}
\vspace{-5pt}
\caption{Fitbit service logs: Proof of login credentials sent in cleartext in a 
HTTP POST request sent from the base to the webserver.
\label{fig:service:log2}}
\vspace{-10pt}
\end{figure}


\section{FitBite: Attacking Fitbit}
\label{sec:fitbite}

We first describe two vulnerabilities of Fitbit, followed by details of the
attacks we have deployed to exploit these vulnerabilities.

\subsection{Vulnerabilities}

\noindent
{\bf Cleartext login information.}
During the initial user login via the Fitbit client software, user passwords
are passed to the website in cleartext (as part of POST data) and then stored
in the log files. Figure~\ref{fig:service:log2} shows a snippet of captured
data, with the cleartext authentication credentials emphasized.

\noindent
{\bf Cleartext HTTP Data processing.}
When syncing data to the website, no data protection/authentication is used --
all requests are sent over plain HTTP.
Capturing tracker data and injecting data into trackers and social
network accounts becomes thus possible.

\subsection{The FitBite Tool}

We have built FitBite, a suite of tools that exploit the above vulnerabilities
to attack Fitbit. FitBite consists of two modules. The Base Module (BM) is used
to retrieve data from the tracker, inject false values and upload them into the
account of the corresponding user on the webserver. The Tracker Module (TM) is
used to read and write the tracker data.  FitBite implements the
following attacks.

\noindent
{\bf Tracker Private Data Capture (TPDC).}
FitBite uses the TM module to discover any tracker device within a radius of 15
ft and capture the fitness information stored on the tracker.  This attack can
be launched in public spaces, particularly those frequented by Fitbit users
(e.g., parks, sports venues, etc).


\begin{figure}
\begin{center}
\includegraphics[width=2.7in]{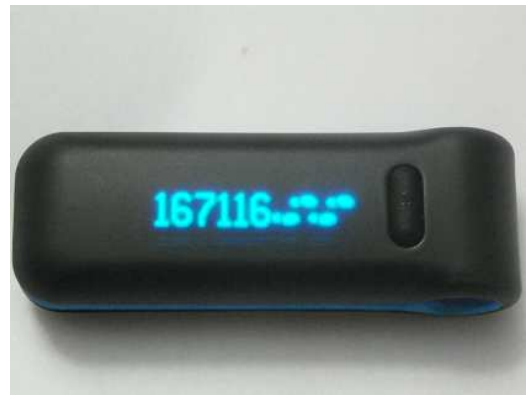}
\caption{Outcome of Tracker Injection (TI) attack on
Fitbit tracker.}
\label{fig:fitbit:pic3}
\end{center}
\end{figure}

\noindent
{\bf Tracker Injection (TI) Attack.}
FitBite uses the TM module along with knowledge of the data and memory bank
formats and required opcode instructions to modify any of the ``real-time''
fitness data stored on neighboring trackers. FitBite allows the attacker to
choose the data to be modified. It then reads the data from the storing memory
bank and modifies the target bytes while keeping the remaining locations
unmodified. The TM can act however modify simultaneously multiple fitness
records (memory banks).

Figure~\ref{fig:fitbit:pic3} shows an example of a victim tracker, displaying
an inflated value for the (daily) number of steps taken by its user. Note that
the tracker's owner (an insider) can also launch this attack.

\begin{figure}
\begin{center}
\includegraphics[width=3.1in]{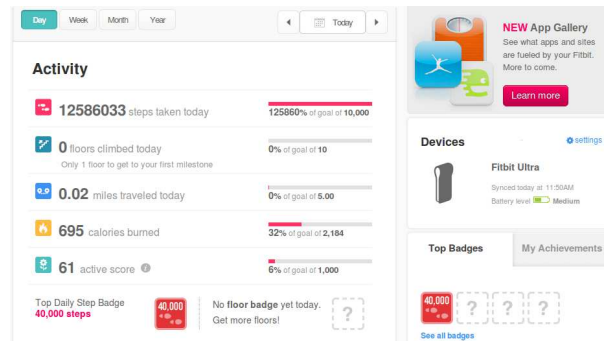}
\caption{Snapshot of Fitbit user account data injection attack.}
\label{fig:attack:inject}
\end{center}
\vspace{-15pt}
\end{figure}

\noindent
{\bf User Account Injection (UAI) Attack.}
Fitbit allows a tracker to report its data to the user's social network account
through any USB base in its vicinity (15 ft. radius).  Specifically, in step
6.b of the Fitbit protocol (see Figure~\ref{fig:fitbit:maincomm}) the base
sends the data to the webserver at the
\url{URL/device/tracker/dumpData/dumpData}.  FitBite enables an attacker to
hijack the data reported by trackers in its vicinity, through the attacker's
corrupt USB base.  FitBite uses the BM module to launch the data injection
attack by fabricating a data reply embedding the desired fitness data (encoded
in the base64 format). The BM sends the reply as an XML block in an HTTP
request to the web server. The webserver does not authenticate the request
message and does not check for data consistency -- thus it accepts the data.


Figure~\ref{fig:attack:inject} shows a snapshot of one account where we have
successfully injected the number of steps taken by the ``account owner'', while
keeping the other values intact. This shows that (i) FitBite can inject an
unreasonable daily step count (12.58 million) into the account of any tracker
owner located in its vicinity and (ii) Fitbit does not check data
consistency -- the 12.58 million steps are shown to correspond to 0.02 traveled
miles.

\noindent
{\bf Free Badges.}
By successful injection of large values in their social networking accounts,
FitBite enables insiders to achieve special milestones and acquire merit
badges, without doing the required work. Figure~\ref{fig:attack:inject} shows
that the injected value of 12.58 million steps, being greater than 40,000,
enables the account owner to acquire a ``Top Daily Step'' badge.

\noindent
{\bf Free Financial Rewards.}
Fitbit users can link their social networking accounts to systems that reward
users for exercising, e.g., Earndit~\cite{Earndit} provides gift cards and
financial prizes. An Earndit user receives 0.75 points for each of her ``very
active'' Fitbit minute and 0.10 points per a ``fairly active'' minute. By keeping
the BM module running and continuously updating the tracker data (once each 15
minutes), FitBite allows an insider to easily record ``fairly active''
minutes. We have created an Earndit account, linked it to one of our Fitbit
accounts and used FitBite to accumulate a variety of undeserved rewards.
Figure~\ref{fig:attack:earndit} shows an example where we have accumulated 200
Earndit points, that can be redeemed to a \$20 gift card.

\begin{figure}
\begin{center}
\includegraphics[width=3.1in,height=1.9in]{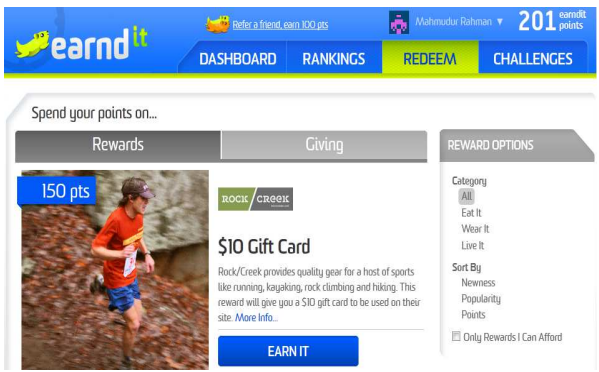}
\caption{Earndit points and available gift cards}
\label{fig:attack:earndit}
\end{center}
\vspace{-15pt}
\end{figure}

\noindent
{\bf Battery Drain Attack.}
FitBite allows the attacker to continuously query trackers in her vicinity,
thus drain their batteries at a faster rate. To understand the efficiency of
this attack, we have experimented with 3 operation modes. First, the
\textit{daily upload} mode, where the tracker syncs with the USB base and the
Fitbit account once per day. Second, the \textit{15 mins upload} mode, where we
kept the tracker within 15 ft. of the base, thus allowing it to be queried once
every 15 minutes. Finally, the \textit{attack} mode, where FitBite's TM module
continuously (an average of 4 times a minute) queried the victim tracker. In
order to not raise suspicions, the BM module uploaded tracker data into the
webserver only once every 15 minutes.

\begin{figure}
\begin{center}
\includegraphics[width=3.1in,height=1.9in]{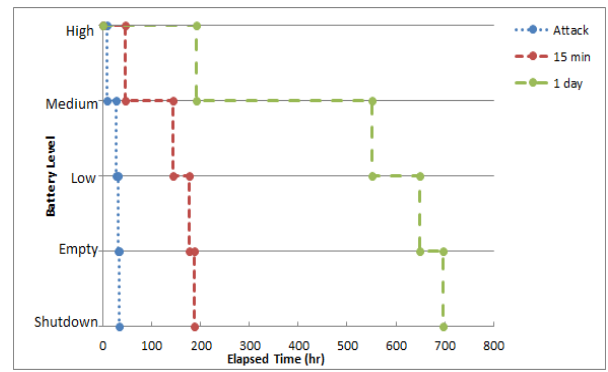}
\caption{Battery drain for three operation modes. The \textit{attack} mode
severely reduced the battery lifetime.}
\label{fig:fitbit:battery}
\end{center}
\vspace{-15pt}
\end{figure}

Figure~\ref{fig:fitbit:battery} shows our battery experiment results for the
three modes. In the daily upload mode, the battery lasted for 29 days.  In the
15 mins upload mode, the battery lasted for 186.38 hours (7 days and 18 hours).
In the attack mode, the battery lasted for a total of 32.71 hours. While this
attack is not fast enough to impact trackers targeted by casual attackers, it
shows that FitBite drains the tracker battery around 21 times faster than the 1
day upload mode and 5.63 times faster than the 15 mins upload mode.

\noindent
{\bf Denial of Service.}
FitBite's injection attack can be used to prevent users from correctly updating
their real-time statistics. A tracker can display up to 6 digit values.  Thus,
when the injected value exceeds 6 digits, the least significant digits can not
be displayed on the tracker. This prevents the user from keeping track of her
daily performance evolution.

\begin{figure*}
\centering
\vspace{-5pt}
\subfigure[]
{\label{fig:attack:rope}{\includegraphics[width=2.3in,height=1.5in]{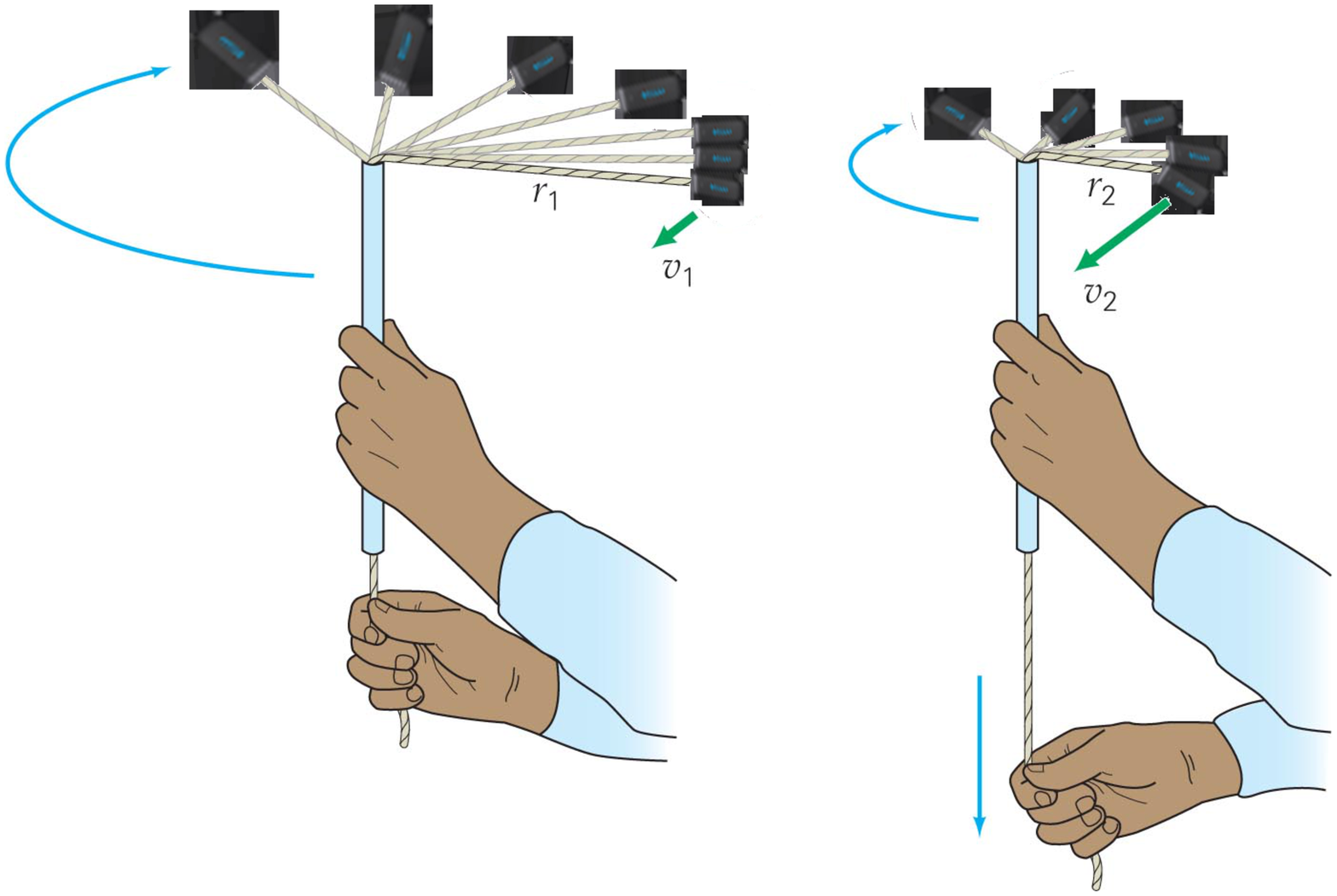}}}
\vspace{-5pt}
\subfigure[]
{\label{fig:attack:physical}{\includegraphics[width=1.9in,height=1.5in]{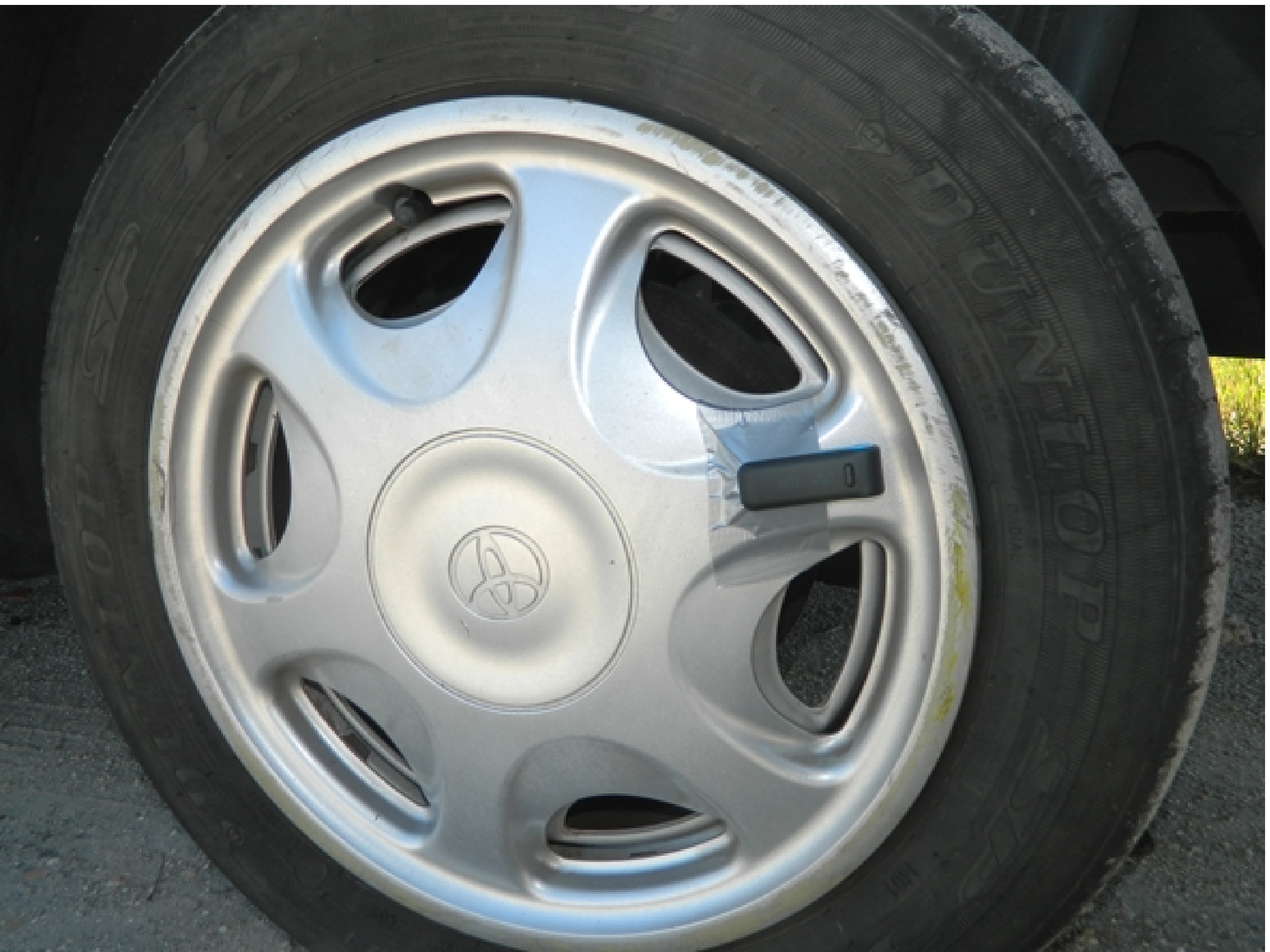}}}
\vspace{-5pt}
\subfigure[]
{\label{fig:attack:wheel}{\includegraphics[width=2.7in,height=1.79in]{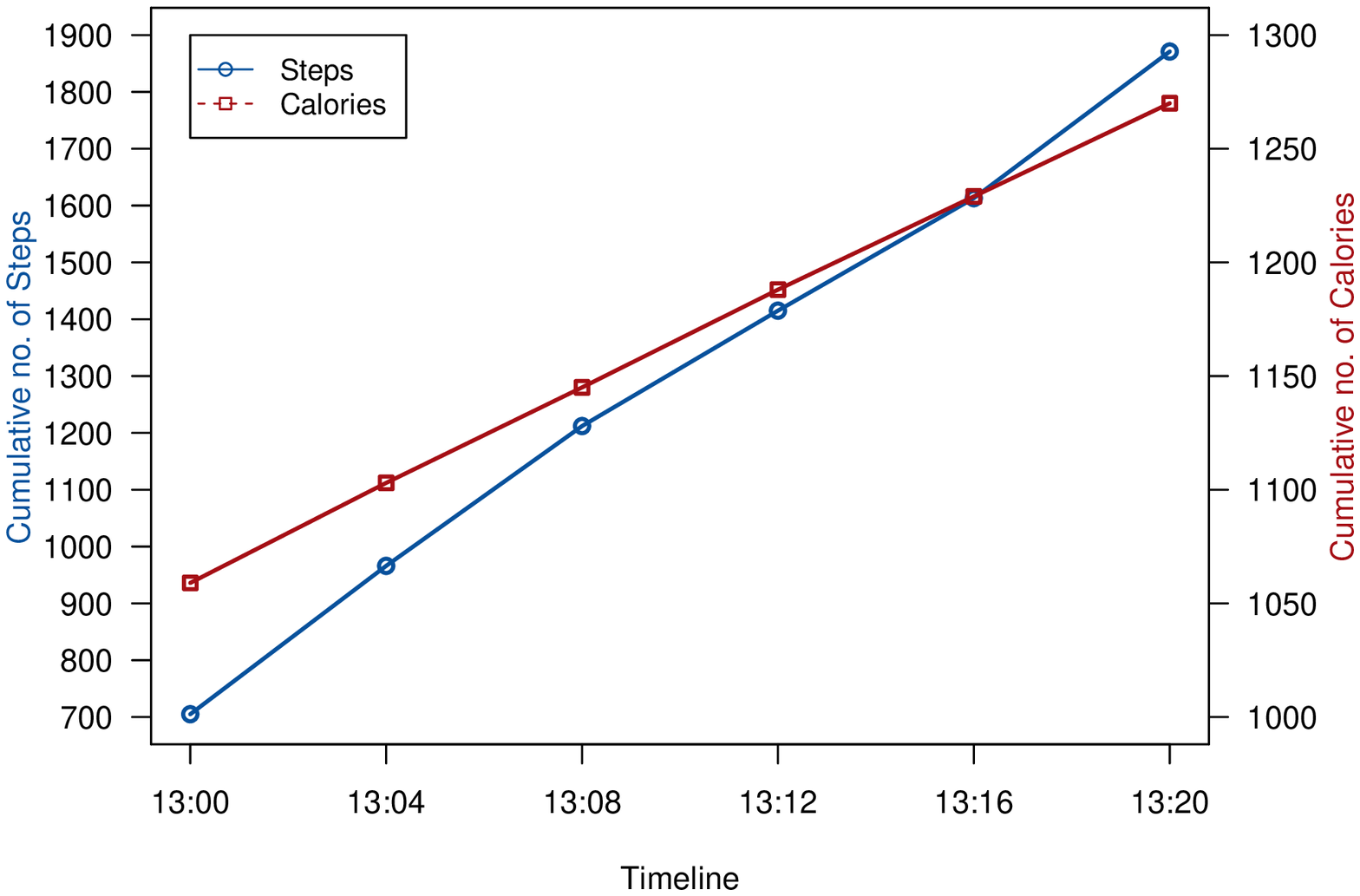}}}
\caption{Physical attack:
(a) Spinning a tracker, when tied with different rope lengths.
(b) Attaching tracker on car wheel.
(c) Effectiveness of wheel attack: evolution of tracker recorded number of steps and calories in time.}
\vspace{-5pt}
\end{figure*}

\noindent
{\bf Mule Attacks.}
Besides attacks exploiting Fitbit's unprotected wireless communications,
adversaries may also launch physical, \textit{mule} attacks, by attaching
trackers to various moving objects. This enables the adversary to increase
fitness parameters with significantly less effort than walking. In a first,
\textit{rope} attack, the adversary spins the tracker attached to a rope
(Figure~\ref{fig:attack:rope} illustrates this attack).  Our experiments show
that the step count increase in the rope attack is a function of the rope
length. For a 1ft rope, the step count increases by 1 for each circumvolution;
for a 2ft rope length, the step count increases by 2 per circumvolution.

While the rope attack requires perseverance, in a second, \textit{wheel}
attack, the adversary attaches a tracker to a car wheel. This enables the
attacker to effortlessly increase the recorded step count, distance and calorie
values when driving.  Figure~\ref{fig:attack:physical} shows a picture of our
``testbed''. We have experimented with this setup, by driving the car over
several 20 minutes sessions.  Figure~\ref{fig:attack:wheel} shows the outcome
of our experiment, in terms of the number of steps and the calories recorded by
the tracker as a function of time, displayed with a 4 minute granularity. At
the end of the experiment, the tracker recorded 1166 steps, 211 calories and
0.9 miles (1.44km).

We used the ``My Tracks'' Android application~\cite{Mytrack} which relies on
GPS readings to measure the average speed and distance traversed by the car.
At an average speed of 16.53kmh, the tracker increases the step count by
approx.  58 steps per minute. The tire type is P17565R14, with a radius of
11.48 inches and a circumference of 72.12 inches ($\approx$ 1.83m). The tracker
was placed 5.85 inches apart from the center of the tire. Thus, for the
circular path of the tracker, the circumference is 36.76 inches ($\approx$
0.93m).

The average distance covered by the car for each 20 minutes session was 5.51km,
when the tire rotated 5510/1.83=3010 times. Given the tracker's circumference
(0.93m), the tracker actually travels 2.8km for the 3010 times tire rotations.
The running stride length of the user bonded with the test tracker is approx.
0.9m, which converts the 2.8km into 3111 steps.  While these values (2.8km/3111
steps) are inconsistent with the 1.44km/1166 steps recorded by the tracker, we
note that the tracker's values are consistent: the tracker converts the steps
into distance according to the user stride length.

%
%
%

\section{FitLock: Protecting Fitbit}
\label{sec:fitlock}

A good Fitbit solution needs to (i) protect against internal and external
attackers, by authenticating the system participants, ensuring the
confidentiality, integrity and freshness of system information, and preventing
denial of service attacks, while simultaneously (ii) taking into consideration
the extreme resource limitations of Fitbit trackers. In this section we
introduce FitLock, a solution that secures the Fitbit system and is efficient in
terms of the imposed computation, storage and communication overheads.

\subsection{The Solution}

FitLock consists of a \textit{bind} procedure (BindUserTracker), where the user
associates a new tracker to her online social network account and an
\textit{upload} procedure (UploadData), where the tracker reports information
upon demand from the social network. Each tracker T has a unique serial number
$id_T$ and a secret symmetric encryption key $sk_T$, shared with the webserver.
These values are stored in a write-once-read-many (WORM) area of the tracker's
memory banks. The tracker never reveals (e.g., displays or communicates) the
secret key.  The webserver stores a database $Map$ that associates a tracker id
to tracker related data, including symmetric key, user id and session id.
Initially, $Map$ only maps tracker ids into corresponding symmetric encryption
keys.

Let $Id_A$ denote the unique user id of the account that user A has on the
Fitbit social network.  In the following, we use the notation
$F(P_1(args_1),..,P_n(args_n))$ to denote a protocol $F$ running between
participants $P_1$,..$P_n$, each with its own input arguments.  For instance,
the following BindTrackerUser protocol involves user A, with her account id and
a time interval $s$ as input arguments, her tracker T, with its id and secret
key as arguments, her base B with no arguments and the (Fitbit) webserver WS,
with its $Map$ structure as input argument. The BindTrackerUser protocol allows
user A to bind her new tracker to her social network account (illustrated in
Figure~\ref{fig:fitlock:bind}).

{\bf BindTrackerUser}(A($Id_A$,s),T($id_T,sk_T$),B(),WS(Map)).
User A logs in into her account on the Fitbit social network (step 1 in
Figure~\ref{fig:fitlock:bind}). A presses T's switch button for $s$
seconds (step 2). Upon this action, the tracker T reports its identifier $id_T$
in cleartext to WS, through the user's base (step 3). WS uses the Map
structure to retrieve the symmetric key associated with the $id_T$, i.e., $sk_T$
(step 4). It then generates a 6 digit long random value, $N$ (step 5). WS sends
to T the request value
\[
id_T, E_{sk_T}(``WS'', Time, N),
\]
where $Time$ is WS's current time (step 6). WS keeps track of all requests sent to
trackers and pending responses, indexed under the tracker id and the nonce
value.  WS associates an expiration time with each entry, and removes entries
as they expire without being answered.

Upon reception of this message, T uses its symmetric key, $sk_T$, to decrypt
it. It verifies the freshness (the $Time$ value) and authenticates WS through
its ability to have encrypted this message, containing the string ``WS'', using
the key $sk_T$.  If the verifications succeed, the tracker displays the 6 digit
random nonce $N$ (step 7). User A reads and enters this nonce into a
confirmation box in her Fitbit social network account (step 9). Then, if WS
finds any pending (not expired) request matching the value entered by the user,
WS associates $Id_A$ to $id_T$ and $sk_T$ in the $Map$ structure (step 10). WS
removes this request from the list of pending requests.


The following procedure, UploadData, is used to secure the Fitbit communication
protocol described in Section~\ref{sec:fitbit}. It involves a tracker T (taking
as arguments its id $id_T$, secret key $sk_T$, stored fitness $data$,
expiration interval $\delta$t and retry counter $r$), a base B (with no
arguments) and the webserver WS (with its $Map$ structure and the same
expiration intervals and counter as T).

All communication between T and WS is encrypted with their shared key $sk_{T}$.
Each communication session between WS and T has a monotonically increasing
session id $S_{wst}$. T and WS do not accept messages with older session id
numbers.

\begin{figure}
\begin{center}
\includegraphics[width=2.7in]{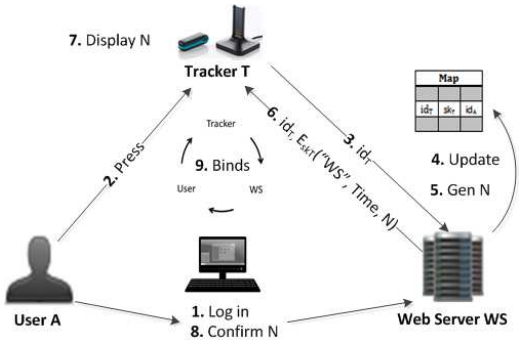}
\caption{The BindTrackerUser protocol between the user, tracker and the (Fitbit)
webserver}
\label{fig:fitlock:bind}
\end{center}
\end{figure}

\noindent
{\bf UploadData}(T($id_T$,$sk_T$,data,$\delta$t,r),B(),WS(Map,$\delta$t,r)).  A
new session starts only after the tracker's beacon is received by the base and
the base sets up a connection with the tracker (step 1).  Within each session,
the communication between WS and T starts with a request from WS followed by a
response from T. Each request contains a request type REQ $\in$ \{TRQ-REQ,
READ-TRQ, WRITE, ERASE, CLOSE\} (see Figure~\ref{fig:fitbit:maincomm}), and a
counter $C_{ws}$ encoding the number of times this particular request has been
re-transmitted.  Within a session, T stores the latest $C_{ws}$ received from
$WS$ for any request type, or -1 if no request has been received yet.  Thus, a
request from WS to T has the format
\[
id_T, E_{sk_T}(REQ,S_{wst},C_{ws}),
\]
where $S_{wst}$ is the current session id and $C_{ws}$ is set to 0 for the
first transmission of the current REQ type. Upon receiving such a message, the
base B uses $id_T$ to route the packet to the correct tracker T in its
vicinity.  T uses its secret key to decrypt the packet and authenticate WS:
verify that the first field is a meaningful request type, the second field
contains the current session id and the value of the third field exceeds its
currently stored value for REQ. If either verification fails, T drops the
packet. Otherwise, T stores the received $C_{ws}$ value, associated with the
REQ type for the current session, and replies to this request with
\[
id_T, E_{sk_T}(RESP,S_{wst},C_T),
\]
where RESP $\in$ \{TRQ-INFO, TRQ-DATA, CLEAR\} denotes T's response type (see
Figure~\ref{fig:fitbit:maincomm}) and $C_T$ is its counter (initialized to 0).

WS waits a predefined interval $\delta$t to receive the reply RESP from T. If
it does not receive it in time, WS repeats the request, with an incremented
counter $C_{ws}$.
%
%
If WS's re-transmission counter reaches a maximum value, $r$, and no
corresponding RESP is received within the $\delta$t interval, WS increments the
session id $S_{wst}$.  Similarly, if C's re-transmission counter reaches the
maximum value $r$ and the next request is not received from WS, T increments
$S_{wst}$. This means that T and WS consider themselves to have been
disconnected and their next communication needs to start from the beginning
(step 1 of Figure~\ref{fig:fitbit:maincomm})) with a new session id.  If T
receives a REQ from WS that has a session id larger (by 1) than its current
session id, T drops the data associated with the current session, and begins a
new session with the incremented session id.

At the successful completion of a session, both T and WS increment the session
id $S_{wst}$. WS stores this value in Map indexed under $id_T$.




\subsection{Data Consistency}

As mentioned in Section~\ref{sec:model}, there exists a strong relationship
between the different activity parameters tracked by Fitbit. However, as
demonstrated by the UAI attack in Section~\ref{sec:fitbite}, Fitbit does not
verify the consistency of the data reported by trackers. FitLock addresses this
vulnerability: Whenever new user data is uploaded on the webserver, FitLock
uses the walk/run stride length and the BMR values to verify the relations
between the number of steps (walking and running) and the distance traversed
and the calories burned by the user. If the relations do not hold (including an
error margin), FitLock considers that the data has been victim of an injection
attack.

\subsection{Analysis}

We now prove several properties of FitLock.

\begin{thm}
Without physical access to the tracker, an attacker cannot hijack the tracker
during the $BindTrackerUser$ procedure.
\end{thm}

\begin{proof}
A \textit{tracker hijack} attack, takes place during a normal execution of the
$BindTrackerUser$ procedure by a victim user for her tracker T. The adversary
attempts to bind the victim tracker T to another user account, potentially
controlled by the attacker.  Let $M$ denotes the Fitbit account owned by the
adversary.  Without physical access to the tracker, the adversary cannot read
the 6 digit random nonce displayed on the tracker and upload it in $M$.

However, the adversary is able to capture packets exchanged by WS and T during
a $BindTrackerUser$ procedure. The adversary could then attempt launch a
\textit{rush} attack. In a rush attack, the adversary decrypts a captured
packet, recovers the nonce $N$ sent by WS to T, and uploads it in $M$,
before the valid user.

Rush attacks are prevented by the semantic security of the encryption scheme of
FitLock -- the adversary cannot recover the nonce.
\end{proof}

FitLock prevents the TPDC attack through the use of semantically secure
encryption. The non-malleability of the encryption also prevents injection
TI, UAI and ensuing free badge and financial rewards
attacks, generated from previously captured (encrypted) messages. The use of
session identifiers and re-transmission counters prevents replay attacks.

\begin{thm}
FitLock prevents DoS and Battery Drain attacks.
\end{thm}

\begin{proof}(Sketch)
FitLock's use of semantically secure symmetric encryption to protect
communications, prevents attackers from obtaining a response from trackers.
The attacker cannot replay requests with old session ids or old counters (for
the current session id): Upon receiving invalid requests or requests with old
session ids or old counter values, the tracker drops them, thus does not
consume power to answer them.
%
%
\end{proof}

\noindent
{\bf Thwarting mule attacks.}
The sensors present in Fitbit trackers are insufficient to prevent insider, mule
attacks: the adversary has control over the tracker and the step count recorded
by the tracker is consistently converted into distance and calorie values. We
propose however two defenses against this attack, relying on the addition of
new sensors on Fitbit trackers.

A first solution relies on GPS chips installed in Fitbit trackers. Broadcom
offers BCM4752~\cite{BCM}, an inexpensive, energy efficient (it uses 50\% less
power than equivalent receivers), small (takes up nearly half the size of
comparable chips) and performant -- delivers 10 times more accurate readings,
and works indoors. By comparing the distance recorded by the GPS receiver
against the distance recorded by the tracker, we can discover inconsistencies
both for rope and wheel attacks. During a rope attack, the GPS location does
not change. During a wheel attack the GPS location changes too much compared
to the number of recorded steps.

A second solution relies on the inclusion of a small heart-rate monitor (HRM)
which can be well suited inside the Fitbit tracker (e.g., the AD8232
AFE~\cite{HRM} which comes in a 4 $\times$ 4-mm sized package).  The HRM device
measures cardiovascular electrical signals from the heart and the tracker only
records the user's activity if it gets such signals from the user. It ensures
that the user is actually wearing the tracker, thus trivially preventing rope
or wheel attacks.

%
%
%
%
%

\section{Evaluation}
\label{sec:evaluation}

\subsection{Experimental Setup}

\begin{figure}
\begin{center}
\includegraphics[width=3.2in]{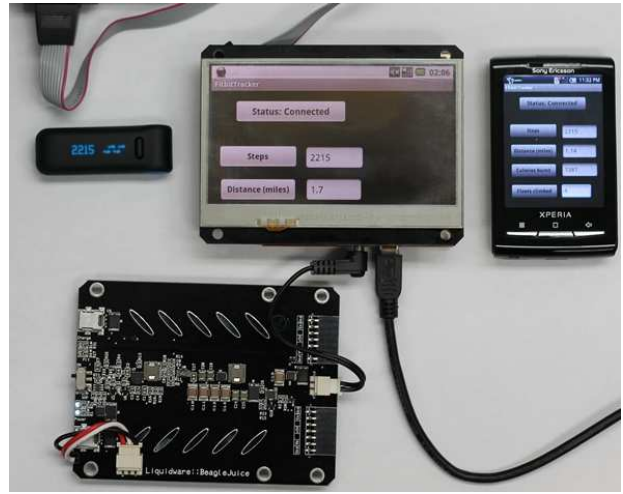}
\caption{Snapshot of testbed for FitLock, consisting of BeagleBoard and Xperia
devices used as Fitbit trackers.}
\label{fig:fitbit:structure}
\end{center}
\vspace{-15pt}
\end{figure}

We implemented FitLock in Android. We have tested the tracker side of FitLock
on a Revision C4 of the BeagleBoard~\cite{Beagle} and an Xperia smartphone. The BeagleBoard
uses the OMAP 3530 DCCB72 720 MHz version and uses a 4500 mAh Li-ion battery to
power our system through a special, 2-pin barrel jack. The Sony Ericsson Xperia
X10 mini smartphone features an ARM 11 CPU @600 MHz and 128MB RAM with Android
OS Eclair 2.1. Similar to the Fitbit tracker, the Xperia device supports ANT+.
In addition, we have used two Dell laptops, one equipped with a 2.4GHz Intel
Core i5 and 4GB of RAM, was used for the web server (built on the Apache web
server 2.4) and the other, equipped with a 2.3GHz Intel Core i5 and 4GB of RAM,
was used for the base.

We implemented a client-server Bluetooth~\cite{Bluetooth} socket communication
protocol between the tracker (Xperia smartphone) and the base using
PyBluez~\cite{Pybluez} python library. In PyBluez, each device acts as a server
and other connected devices act as clients in P2P communications. For
connectivity between the base and the webserver, the laptops use their own
802.11b/g Wi-Fi interfaces. Figure~\ref{fig:fitbit:structure} shows a snapshot
of our testbed.

For encryption we experimented with RC4~\cite{RC4}, AES~\cite{AES} and the 
Salsa20~\cite{Salsa20} stream cipher, selected in the final eSTREAM portfolio
~\cite{ESTREAM}. The 20-round Salsa20 is built on a pseudorandom function based 
on 32-bit addition, bitwise addition (XOR) and rotation operations. It uses a 
256-bit key, a 64-bit nonce, and a 64-bit stream position to a 512-bit output.
As of 2012, there are no published attacks on the full Salsa20/20; the best 
attack known~\cite{JSSWC} breaks 8 of the 20 rounds.

\subsection{Results}

\begin{figure}
\centering
\includegraphics[width=2.5in]{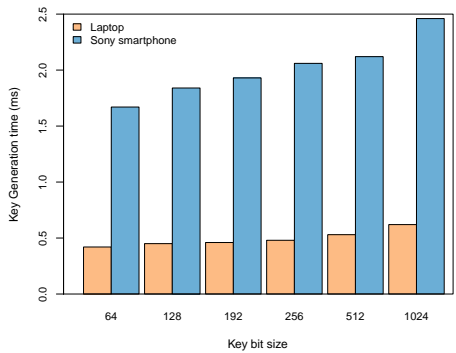}
\caption{Secret key generation time overhead on Xperia device and laptop.}
\label{fig:fitlock:key}
\end{figure}

In the following, all reported values are averages taken over at least 10
independent protocol runs.

\noindent
{\bf Key generation overhead.}
We have first measured the overhead of generating (AES) secret keys.
Figure~\ref{fig:fitlock:key} shows the overhead on the laptop and the Xperia
device, when the key bit size ranges from 64 to 1024 bits. Note that even a
resource constrained smartphone takes only 2.46 ms to generate 1024 bit keys
(0.62ms on the laptop).

\begin{figure*}
\centering
\vspace{-5pt}
\subfigure[]
{\label{fig:fitlock:sony}{\includegraphics[width=2.3in,height=1.9in]{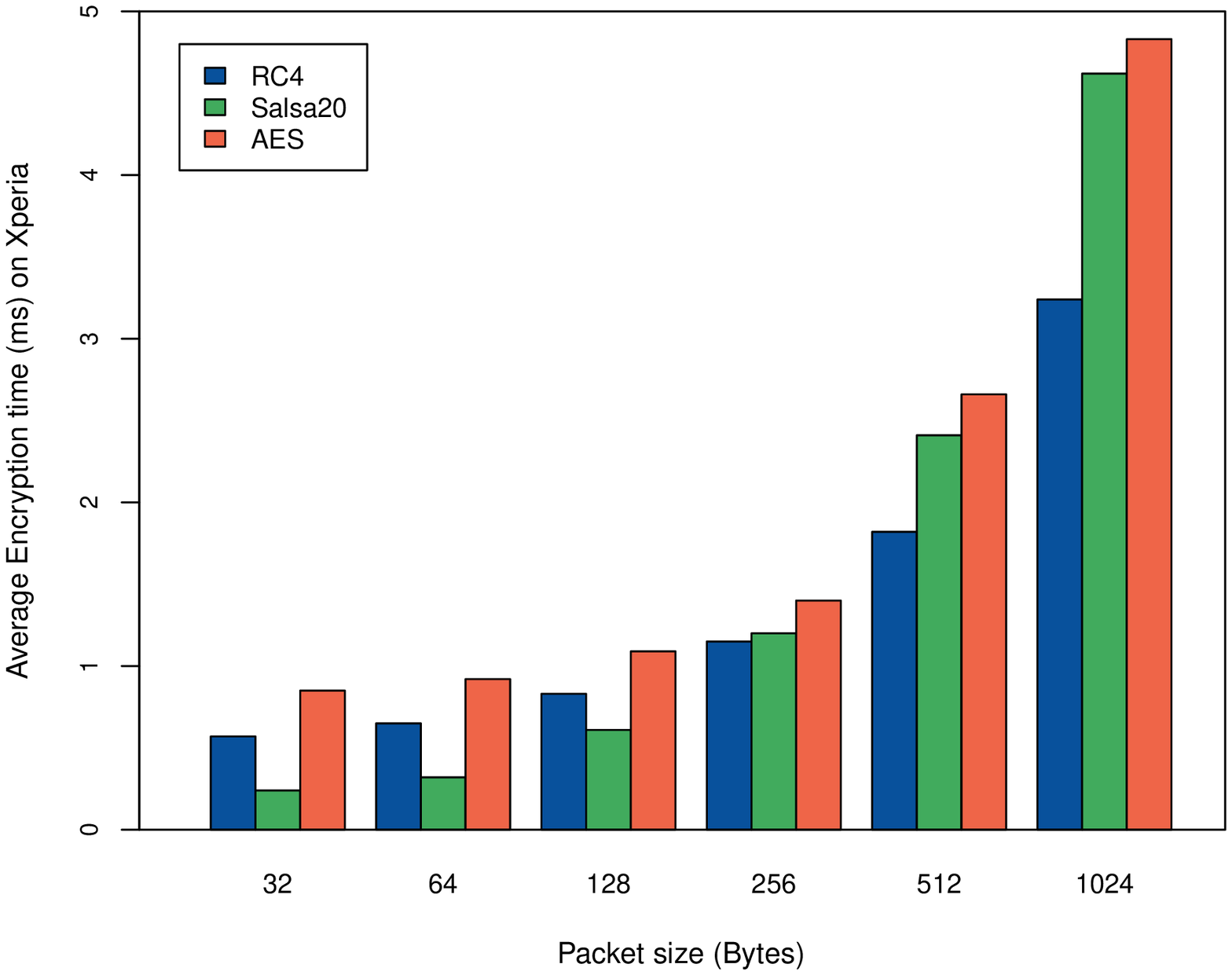}}}
\vspace{-5pt}
\subfigure[]
{\label{fig:fitlock:compare}{\includegraphics[width=2.3in,height=1.9in]{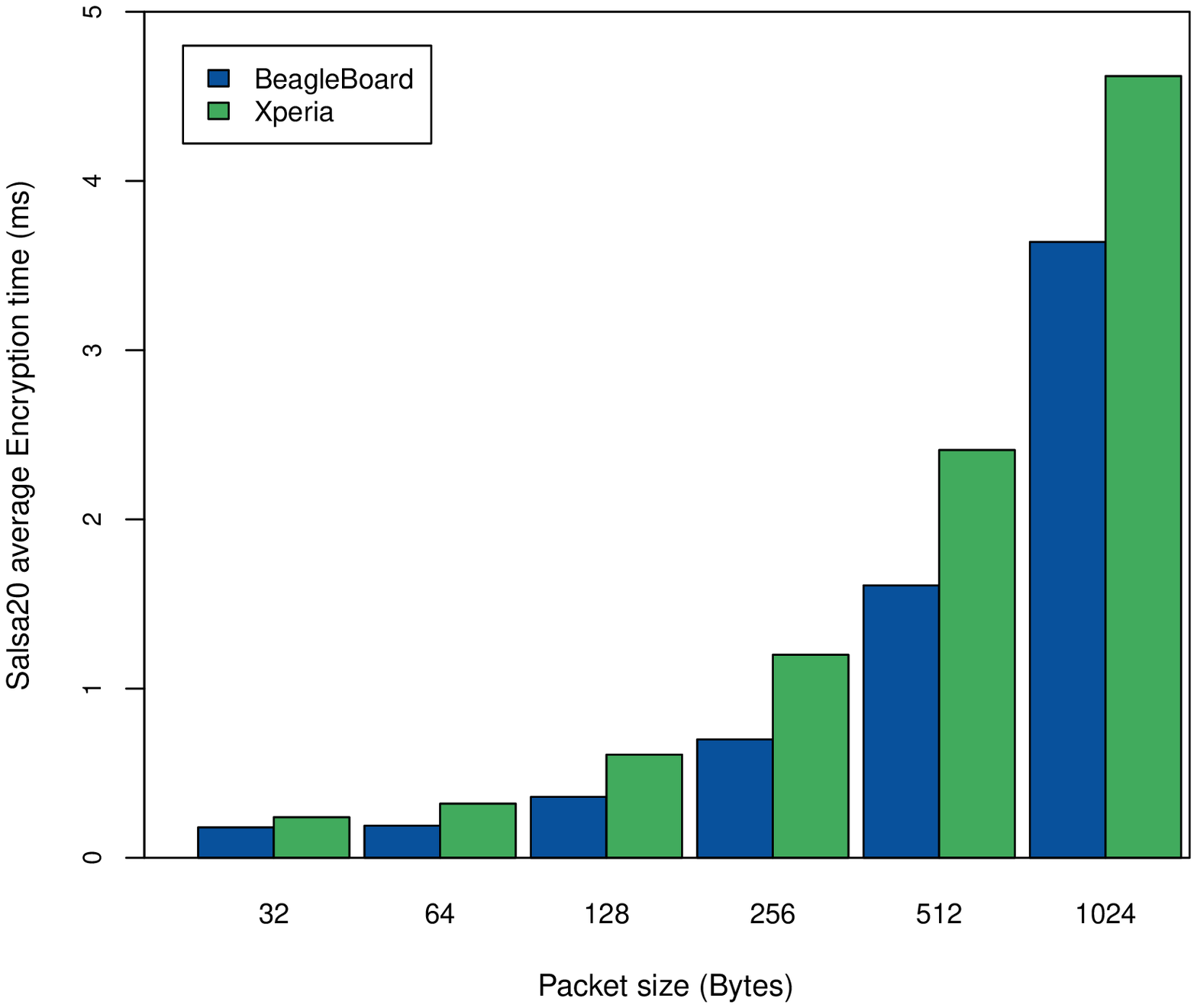}}}
\vspace{-5pt}
\subfigure[]
{\label{fig:fitlock:decrypt}{\includegraphics[width=2.3in,height=2.2in]{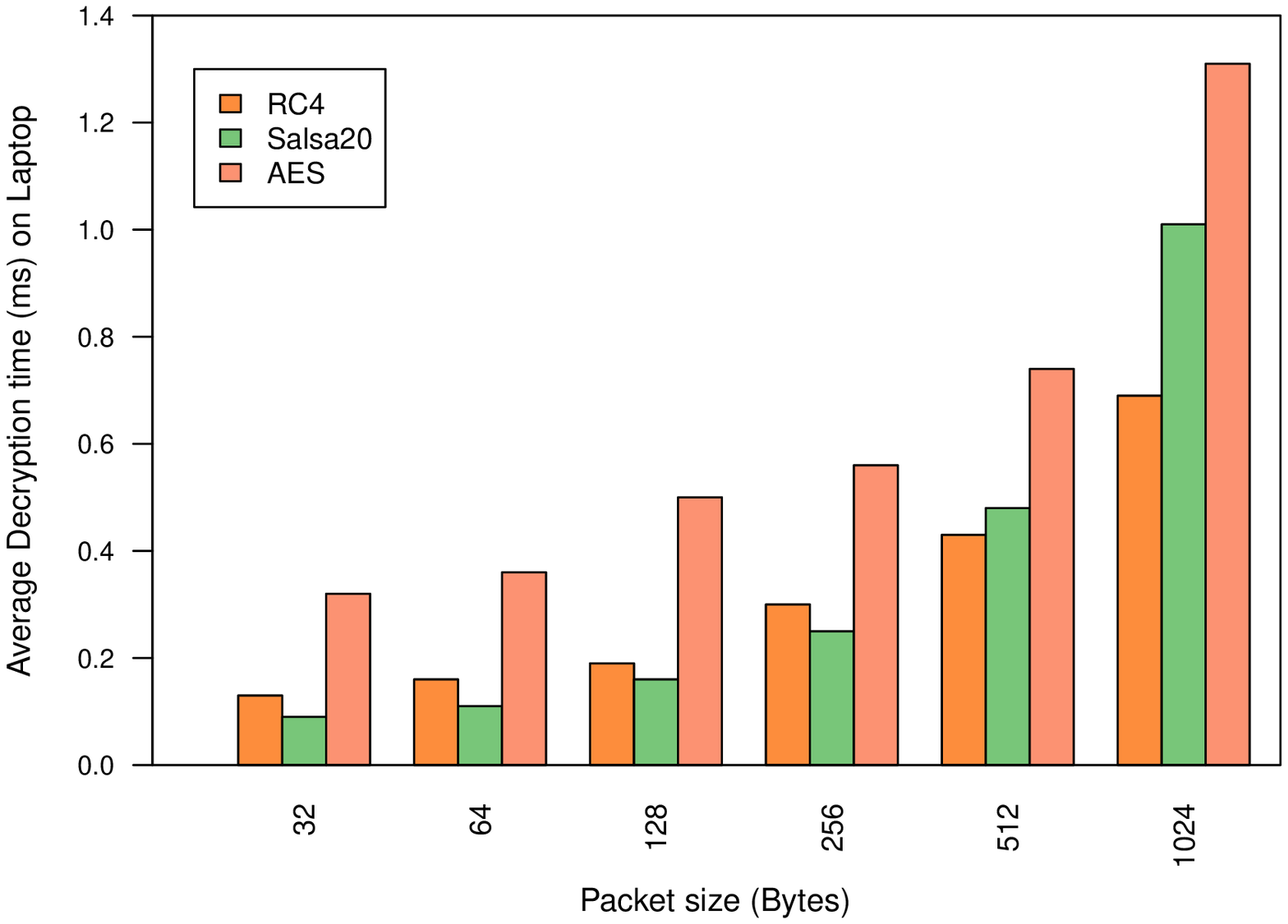}}}
\caption{FitLock overhead.
(a) Encryption time overhead on Xperia.
(b) Comparison of Salsa20 encryption time overhead when tested on BeagleBoard and Xperia.
(c) Decryption time overhead on webserver (Dell laptop).}
\vspace{-15pt}
\end{figure*}

\noindent
{\bf FitLock overhead on tracker.}
A potential bottleneck of FitLock is in the encryption of packets by the
tracker. In order to verify if this can be the case, we compared the
performance of RC4, Salsa20 and AES. We set the key size to 128 bits. We ran
these protocols both on the BeagleBoard and the Xperia while the packet size
ranges from 32 bytes to 1024 bytes.  Figure~\ref{fig:fitlock:sony} shows the
execution time of the three protocols on the Xperia smartphone. For small
packet sizes, Salsa20 performs the best. As the packet size increases, RC4
performs slight better than Salsa20.  Both RC4 and Salsa20 outperform AES for
any packet size. Even for a packet size of 1024 bytes, the average encryption
times for RC4, Salsa20 and AES are only 3.24ms, 4.62ms and 4.83ms respectively.
In Figure~\ref{fig:fitlock:compare}, we compared the encryption overhead of
Salsa20 when running on the BeagleBoard and on the Xperia smartphone. The
BeagleBoard performs better due to its more powerful CPU: it takes only 3.64ms
to encrypt 1024 bytes packets. Thus, both the BeagleBoard and the Xperia device
are able to generate hundreds of packet encryptions per second, making
encryption an unlikely source of bottlenecks.

\noindent
{\bf FitLock overhead on webserver.}
We further examined the packet decryption overhead on the webserver using the
above mentioned protocols. Figure ~\ref{fig:fitlock:decrypt} shows the
dependence of the decryption time on the packet size. RC4 and Salsa20 perform
better than AES. Even for 1024 byte packets, the average decryption overheads
for RC4, Salsa20 and AES are 0.69ms, 1.01ms and 1.31ms respectively.

\begin{figure}
\begin{center}
\includegraphics[width=2.3in]{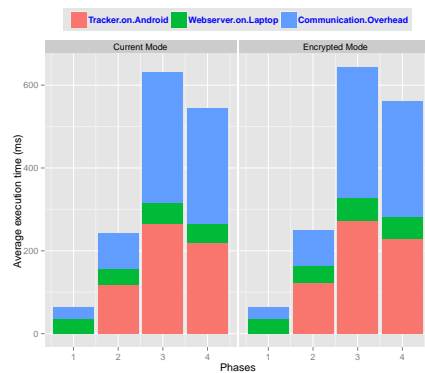}
\caption{Comparison of end-to-end delay between the current Fitbit solution 
and our proposed encrypted solution}
\label{fig:fitbit:end2end}
\end{center}
\vspace{-15pt}
\end{figure}

\noindent
{\bf End-to-end FitLock overhead.}
Finally, we report the measured end-to-end performance of FitLock and compare
it against the performance of Fitbit. We have implemented and tested both
Fitbit and FitLock on our testbed. Figure~\ref{fig:fitbit:end2end} shows our
results split into the times of each of the 4 phases of the
webserver-to-tracker communication protocol described in
Figure~\ref{fig:fitbit:maincomm}.  We have set the secret key size to 256 bits.
The end-to-end (sum over all 4 phases) time of the FitLock protocol is 1518ms.
The total time of Fitbit is 1481ms. Thus, FitLock adds an overhead of 37ms,
accounting for 2.4.\% of Fitbit's time.


\section{Related Work}
\label{sec:related}

%
%
%


Halperin et al.~\cite{Pacemaker} demonstrated attacks on pacemakers and
implantable cardiac defibrillators, and proposed zero-power defenses. 
Similarly, Li et. al.~\cite{Insulin} demonstrated successful security 
attacks on a commercially deployed glucose monitoring and insulin delivery 
system and provided defenses against the proposed attacks. Proximity-based 
access control~\cite{Proximity} has been proposed as a technique for 
implantable medical devices to verify the distance of the communicating 
peer before initiating wireless communication, thereby limiting attackers 
to a certain physical range. Although similar in overall objectives, 
our work differs significantly in the attack methodologies and proposed 
defenses.


Barnickel et al.~\cite{HealthNet} targeted security and privacy issues for
HealthNet, a health monitoring and recording system. They proposed a security
and privacy aware architecture, relying on data avoidance, data minimization,
decentralized storage, and the use of cryptography. Marti et al.~\cite{Ramon} 
described the requirements and implementation of the security mechanisms for 
MobiHealth, a wireless mobile health care system. MobiHealth relies on 
Bluetooth and ZigBee link layer security for communication to the sensors 
and uses HTTPS mutual authentication and encryption for connections to the 
backend.

Lim et al.~\cite{WBAN} analyzed the security of a remote cardiac 
monitoring system. The data transfer was modeled as starting from the 
sensors, reaching a Body Area network (BAN) gateway, then a wireless router 
and through the Internet to a final monitoring server.
Muraleedharan et al.~\cite{Muraleedharan} proposed two types of possible
denial-of-service attacks including Sybil~\cite{Sybil} and
wormhole~\cite{Wormhole} attacks in a health monitoring system using wireless
sensor networks. They further proposed an energy-efficient cognitive routing
algorithm to deal with those attacks.

Sriram et. al.~\cite{PervasiveData} took an in-depth look at potential
health-monitoring usage scenarios and highlighted research challenges required
to ensure and assess quality of sensor data in health-monitoring systems.
The work of Stanford~\cite{PervasiveSec} stresses the need of meeting stringent
privacy and security requirements, especially to protect confidential medical
records and the organizations and end users that employ them.

%
%

\noindent
{\bf Bottom line.}
Most related work either (i) proposes a novel system with embedded defense
mechanisms to handle security and privacy issues or (ii) introduces different
attacks against the wireless network communication in a health monitoring
system.  Our work not only provides hands-on attacks for a popular fitness and
healthcare tracking system but also provides efficient end-to-end defenses and
proves their efficacy in preventing and thwarting attacks. 

\section{Conclusions}
\label{sec:conclusions}

In this paper, we discussed security and privacy issues related to a renowned 
and widely accepted and used fitness tracking system. We showed that through 
reverse engineering of the ANT protocol and data communication, both passive 
and active attacks can be launched on the system using off-the-shelf 
software module. We then analyzed the various attack scenarios and also 
proposed various types of possible defenses against them. We believe that
our proposed attack methodology and defenses may be applicable to several 
wearable and implantable healthcare systems. Healthcare appliance security 
is a critical challenge that demands the immediate attention of the research 
community.

\bibliographystyle{unsrt}
\bibliography{tracker,health,osn}

\end{document}